\documentclass[11pt]{article}
\usepackage{amsfonts}
\usepackage{color}
\usepackage{amsmath}
\usepackage{revsymb4-2}
\usepackage{fullpage}
\usepackage{amssymb}
\usepackage{graphicx, caption, subcaption}
\usepackage{hyperref}
\usepackage{float}
\usepackage{soul}
\usepackage{bbding}
\usepackage{lscape}
\usepackage{ulem}
\usepackage{stmaryrd}

\usepackage[linecolor=red,backgroundcolor=red!25,bordercolor=red,prependcaption,textsize=tiny]{todonotes}

\providecommand{\U}[1]{\protect\rule{.1in}{.1in}}

\newtheorem{theorem}{Theorem}

\newtheorem{conjecture}[theorem]{Conjecture}

\newtheorem{example}[theorem]{Example}

\newtheorem{observation}[theorem]{Observation}
\newtheorem{proposition}[theorem]{Proposition}
\newtheorem{remark}[theorem]{Remark}

\newenvironment{proof}[1][Proof]{\noindent\textbf{#1}\ \ }{\hfill\rule{0.5em}{0.5em}}
\def\tr{\operatorname{tr}}

\def\1{\openone}

\newcommand{\ket}[1]{|#1 \rangle}
\newcommand{\bra}[1]{\langle #1|}

\let\originalleft\left
\let\originalright\right
\renewcommand{\left}{\mathopen{}\mathclose\bgroup\originalleft}
\renewcommand{\right}{\aftergroup\egroup\originalright}

\usepackage{color}

\newcommand{\be}{\begin{equation}}
\newcommand{\ee}{\end{equation}}
\newcommand{\bea}{\begin{eqnarray}}
\newcommand{\eea}{\end{eqnarray}}

\begin{document}

\title{On the effectiveness of Bayesian discrete feedback for quantum information reclaiming}

\author{Milajiguli Rexiti$^{1,2}$, Samad Khabbazi Oskouei$^3$, and Stefano Mancini$^{1,2}$\\ \\
\small $^1$School of Science and Technology, University of Camerino,\\
\small Via Madonna dell Carceri 9, 62032 Camerino, Italy\\
\small $^2$INFN, Sezione di Perugia, Via A. Pascoli, 06123 Perugia, Italy\\
\small $^3$Department of Mathematics, Varamin-Pishva Branch, Islamic Azad University, \\
\small Varamin 33817-7489, Iran
}

\maketitle

\begin{abstract}
We consider discrete time feedback aimed at reclaiming quantum information after a channel action.
We compare Bayesian and Markovian strategies. We show that the former does not offer any advantage for qubit channels, while its superior performance can appear in higher dimensional channels. This is witnessed by cases study for qutrit channels.
\end{abstract}


\section{Introduction}

Controlling quantum systems becomes of paramount importance with the advent of quantum technologies.
The theory of feedback control shares many analogies with the theory of error correction.
A measurement is performed to get information (likewise error syndrome is extracted) and then actuation is performed (likewise error recovery is carried out).
These become particularly evident when employing feedback for coherence recovery. To this end the prototypical model
was proposed in Refs.\cite{GW03,GW04}. The idea was that by having access to the environment, the unitary part of each error can be corrected.
In Ref.\cite{MCM11} it was put forward that the shceme needs an optimization also on the measurements.

The multistep version leads to the possibility of considering Marokovian or Bayesian realizations.
In the former, differently from the latter, the measurement and correction are performed  at each step without memory of previous steps.

Bayesian feedback is not much explored because of its intrinsic difficulties \cite{RNM20,WMW02}. For it, in the discrete version devoted to coherence recovery, it has been shown (at least for qubit channels) that  the actuation only in the last step is equivalent to the actuation at each step \cite{GW04}. However, the problem of Bayesian vs Markovian feedback was not addressed.

Here, in the context of discrete feedback aimed at quantum coherence recovery, we compare the two strategies and show that Bayesian feedback does not offer any advantage for qubit channels, while its superior performance can appear in higher dimensional channels. This is witnessed by cases study for qutrit channels.


\section{The feedback model}

A quantum channel is a linear, completely positive and trace preserving (CPTP) map on the set ${\mathfrak{S}}(\cal H)$ of density operators over a Hilbert space $\cal H$. Every quantum channel ${\cal T}:{\mathfrak{S}}(\cal H)\to {\mathfrak{S}}(\cal H)$ can be expressed in the operator sum (Kraus) representation:
\begin{equation}\label{kr}
{\cal T}(\rho)=\sum_x T_x\rho T_x^\dag,
\end{equation}
where $T_x:{\cal H}\to{\cal H}$ are linear operators (called Kraus operators) satisfying the normalization condition $\sum_x T_x^\dag T_x=I$. The number of non-zero operators in the Kraus representation is called the Kraus rank and it is known to be at maximum $d^2$, where $d$ is the dimension of $\cal H$.

Every channel can be seen as arising from a unitary interaction between the system of interest and the environment after discarding this latter. We assume here that the environment is accessible and hence measurable.
We denote by $\{x\}$ the measurement outcomes corresponding to the positive operator valued measure (POVM)
$\{T_x^\dag T_x\}$.

Notice that given a channel $\cal T$, its Kraus decomposition is not unique. Hence different Kraus decompositions correspond to different measurements performed on the environment to gather information about the system of interest.

To recover the state which has undergone the channel action, we exploit the information gathered from environment,
hence the feedback scheme.
Actually, we
apply recovery operators $\{R_x\}$ (conditioned to the outcomes $\{x\}$) in such a way that $ \sum_x R_x T_x \rho T_x^\dag R_x^\dag $ is as close as possible to the initial state $\rho$.

Let the polar decomposition of the Kraus operators be $T_x=V_x |T_x|$, with $V_x$ the unitary part of $T_x$. Then,
$V_x$ represents the invertible part of the occurred error $T_x$. Thus we choose the recovery $R_x=V_x^\dag$ \cite{GW03,GW04}. As a consequence, after correction the channel outcome becomes:
\be\label{1app}
{\cal T}_{corr}(\rho)=\sum_x |T_x|\rho |T_x|.
\ee
The figure of merit to evaluate the effectiveness of this feedback scheme is the entanglement fidelity of the channel (see e.g. \cite{MW21})
\be\label{fiddef}
F({\cal T}_{corr}):=\langle\Psi| ({\rm id}\otimes {\cal T}_{corr}) (\Psi) |\Psi\rangle,
\ee
where $|\Psi\rangle\in{\cal H}\otimes{\cal H}$ is a maximally entangled state and $\Psi\equiv|\Psi\rangle\langle\Psi|$.
It turns out that
\be\label{fid}
F({\cal T}_{corr})=\frac{1}{d^2}\sum_x \left( \tr |T_x|\right)^2.
\ee
When this quantity reaches the value one, it means that the identity channel is recovered.

Now, assume that the channel is applied twice and after each application of the channel we use the feedback action.
Then we can have two different scenarios (see Fig.\ref{2fb}): i) at the second step, the measurement result on the environment is fed back to the system in an independent way from the first step (we refer to this -- non adaptive -- strategy as \textit{Markovian feedback}); ii) at the second step, the measurement result on the environment is fed back to the system in a way that also depends on the first step (we refer to this -- adaptive -- strategy as \textit{Bayesian feedback}).

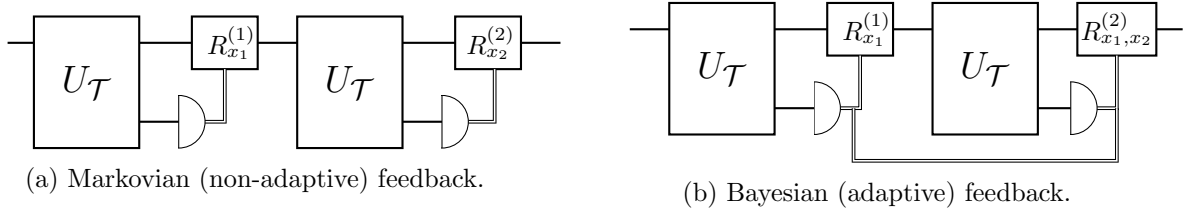
\begin{figure}[H]
	\begin{center}
		\begin{subfigure}[bh]{0.4\textwidth}
		\centering
			\begin{tikzpicture}[scale=0.35]
			\draw[thick] (-1,2) -- (0,2);\draw[thick] (4,2) -- (6,2);\draw[thick] (8.5,2) -- (10,2);\draw[thick] (14,2) -- (16,2);\draw[thick] (18.5,2) -- (20,2);
			\draw[thick](0,3)rectangle(4,-2);\draw[thick](10,3)rectangle(14,-2);
			\draw[thick](6,3)rectangle(8.5,1); \draw[thick](16,3)rectangle(18.5,1);
			\draw[thick] (4,-1) -- (5.5,-1); \draw[thick] (14,-1) -- (15.5,-1);
			\draw[double](6.5,-1)--(7.25,-1)--(7.25,1);
			\draw[double](16.5,-1)--(17.5,-1)--(17.5,1);
				\begin{scope}
    				\clip (5.5,0) rectangle (6.5,-2);
 				\draw (5.5,-1) circle(1);
    				\draw (5.5,0) -- (5.5,-2);
				\end{scope}
				\begin{scope}
    				\clip (15.5,0) rectangle (18,-2);
   		 		\draw (15.5,-1) circle(1);
				\draw (15.5,0) -- (15.5,-2);
				\end{scope}
			\draw(2,0.5) node[font = \fontsize{15}{15}]{$U_{\cal T}$};\draw(12,0.5) node[font = \fontsize{15}{15}]{$U_{\cal T}$};
			\draw(7.5,2) node[font = \fontsize{10}{10}]{$R_{x_1}^{(1)}$};\draw(17.35,2) node[font = \fontsize{10}{10}]{$R_{x_2}^{(2)}$};
			\end{tikzpicture}
 			\caption{Markovian (non-adaptive) feedback.}
      			\label{M1}
		\end{subfigure}
		\quad\quad\quad\quad
		\begin{subfigure}[bh]{0.4\textwidth}
		\centering
			\begin{tikzpicture}[scale=0.35]
				\draw[thick] (-1.5,2) -- (0,2);\draw[thick] (4,2) -- (6,2);
				\draw[thick] (8.5,2) -- (10,2);\draw[thick] (14,2) -- (15.5,2);						
				\draw[thick] (18.5,2) -- (19.5,2);
				\draw[thick](0,3)rectangle(4,-2);\draw[thick](10,3)rectangle(14,-2);
				\draw[thick](6,3)rectangle(8.5,1); \draw[thick](15.5,3)rectangle(18.5,1);
				\draw[thick] (4,-1) -- (5.5,-1); \draw[thick] (14,-1) -- (15.25,-1);
				\draw[double](6.5,-1)--(7.25,-1)--(7.25,1);
				\draw[double](16.25,-1)--(17,-1)--(17,1);
				\draw[double](7,-1)--(7,-3)--(17,-3)--(17,-1);
					\begin{scope}
    					\clip (5.5,0) rectangle (6.5,-2);
    					\draw (5.5,-1) circle(1);
    					\draw (5.5,0) -- (5.5,-2);
    					\end{scope}
					\begin{scope}
    					\clip (15.25,0) rectangle (16.25,-2);
    					\draw (15.25,-1) circle(1);
    					\draw (15.25,0) -- (15.25,-2);
					\end{scope}
			\draw(2,0.5) node[font = \fontsize{15}{15}]{$U_{\cal T}$};\draw(12,0.5) node[font = \fontsize{15}{15}]{$U_{\cal T}$};
			\draw(7.5,2) node[font = \fontsize{10}{10}]{$R_{x_1}^{(1)}$};\draw(17,2) node[font = \fontsize{10}{10}]{$R_{x_1,x_2}^{(2)}$};
			\end{tikzpicture}
 	\caption{ Bayesian (adaptive) feedback.}
      \label{M2}
	\end{subfigure}
\end{center}
\caption{Two-step feedback.}
\label{2fb}
\end{figure}

According to \eqref{1app}, the output after the second action of the channel will be
\begin{align}
{\cal T}\circ {\cal T}_{\rm corr}(\rho)=\sum_{x_2} T^{(2)}_{x_2} \left(\sum_{x_1} \left| T^{(1)}_{x_1} \right| \rho \left| T^{(1)}_{x_1} \right|\right) {T^{(2)}_{x_2}}^\dag,
\end{align}
where, to be general, we considered different sets of Kraus operators $\{T^{(1)}_{x_1}\}$ and $\{T^{(2)}_{x_2}\}$ for the channel ${\cal T}$ in the two steps.
This corresponds to the possibility of performing different kind of measurements on the environment at each step.

Then operating the second feedback step according to the Markovian strategy we will get
\begin{align}
({\cal T}_{\rm corr} \circ {\cal T}_{\rm corr})(\rho)=\sum_{x_1,x_2} \left(
\left| T^{(2)}_{x_2} \right|  \, \left| T^{(1)}_{x_1} \right|\right) \, \rho \, \left( \left| T^{(1)}_{x_1} \right| \, \left| T^{(2)}_{x_2}
\right| \right).
\end{align}
In contrast, operating the second feedback step according to the Bayesian strategy we will get
\begin{align}
\left({\cal T}\circ {\cal T}_{\rm corr}\right)_{\rm corr}(\rho)=\sum_{x_1,x_2} \left(
\left| T^{(2)}_{x_2}  \, \left| T^{(1)}_{x_1} \right|\right|\right) \, \rho \, \left(\left| \left|{T^{(1)}_{x_1}}^\dag \right| \, {T^{(2)}_{x_2}}^\dag \right| \right).
\end{align}
As a consequence of \eqref{fid} we will have the entanglement fidelity for the two cases as
\begin{align}
F({\cal T}_{\rm corr} \circ {\cal T}_{\rm corr})&=\frac{1}{d^2}\sum_{x_1,x_2} \left( {\rm tr} \left|{T}^{(2)}_{x_2}\right|\,
\left|{T}^{(1)}_{x_1}\right|
\right)^2,\\
F\left(\left({\cal T} \circ {\cal T}_{\rm corr}\right)_{\rm corr}\right)&=\frac{1}{d^2}\sum_{x_1,x_2} \left( {\rm tr} \left| {T}^{(2)}_{x_2} \, \left| {T}^{(1)}_{x_1}\right|\right|
\right)^2,
\end{align}
respectively.

Extending the above reasoning to $n$ feedback steps we will arrive to
\begin{align}
F_n\equiv F\left(\underbrace{{\cal T}_{\rm corr}\circ \ldots \circ {\cal T}_{\rm corr}}_{n-\text{times}}\right)
&=\frac{1}{d^2}\sum_{x_1,\ldots,x_n} \left( {\rm tr} \left| T^{(n)}_{x_n}\right| \ldots \left| T^{(1)}_{x_1}\right|
\right)^2,
\end{align}
and
\begin{align}
F'_n\equiv F\left(
\underbrace{\left({\cal T}\circ \ldots \left({\cal T} \circ {\cal T}_{\rm corr}\right)_{\rm corr} \ldots
\right)_{\rm corr}}_{n-\text{times}}
\right)&=\frac{1}{d^2}\sum_{x_1,\ldots, x_n} \left( {\rm tr} \left| T^{(n)}_{x_n} \ldots
\left| T^{(2)}_{x_2} \,
\left| T^{(1)}_{x_1}\right|\right|\right|
\right)^2.
\end{align}

\begin{observation}\label{obs}
It is $F'_n\geq F_n$. This simply follows from the fact that for each $k$ such that $2\leq k \leq n$ it is
\begin{align}
{\rm tr} \left| T^{(k)}_{x_k} \left| T^{(k-1)}_{x_{k-1}} \ldots \right| \right|
={\rm tr} \left| V^{(k)}_{x_k} \left|T^{(k)}_{x_k}\right| \, \left| T^{(k-1)}_{x_{k-1}} \ldots \right| \right|
= {\rm tr} \left| \left| T^{(k)}_{x_k} \right| \, \left| T^{(k-1)}_{x_{k-1}} \ldots \right| \right|
\geq  {\rm tr} \left| T^{(k)}_{x_k} \right| \left| T^{(k-1)}_{x_{k-1}} \ldots \right|.
\end{align}
The second equality follows from the unitary invariance of the 1-norm, while the inequality is consequence of the
property $\sum_i \lambda_i(A) \leq \sum_i |\lambda_i(A)| \leq \sum_i \sigma_i(A)$, with $\lambda_i(A)$ (resp. $\sigma_i(A)$)
denoting the eigenvalues (resp. singular eigenvalues) of a square matrix $A$ \cite{HLA}.
\end{observation}


Then, denoting by $\{\widetilde{T}^{(k)}_{x_k}\}$ the optimal Kraus decomposition at the $k$-th step, it will be
\begin{align}
F_n=\frac{1}{d^2}\sum_{x_1,\ldots,x_n} \left( {\rm tr} \left| \widetilde{T}^{(n)}_{x_n}\right| \ldots \left| \widetilde{T}^{(1)}_{x_1}\right|
\right)^2,
\end{align}
and
\begin{align}
F'_n&=\frac{1}{d^2}\sum_{x_1,\ldots, x_n} \left( {\rm tr} \left| \widetilde{T}^{(n)}_{x_n} \ldots
\left| \widetilde{T}^{(2)}_{x_2} \,
\left| \widetilde{T}^{(1)}_{x_1}\right|\right|\right|
\right)^2.
\end{align}


\begin{observation}\label{obs1}
If it happens that
\begin{equation}
\left[ \left|\widetilde{T}^{(i)}_{x_k}\right| \, , \, \left|\widetilde{T}^{(j)}_{x_l}\right|\right]=0, \qquad \forall i,j \quad \text{and}\quad \forall x_k, x_l,
\end{equation}
then $F_n=F'_n$, because in such a case
$\left| \widetilde{T}^{(n)}_{x_n} \ldots
\left| \widetilde{T}^{(2)}_{x_2} \,
\left| \widetilde{T}^{(1)}_{x_1}\right|\right|\right|
= \left| \widetilde{T}^{(n)}_{x_n}\right| \ldots \left| \widetilde{T}^{(1)}_{x_1}\right|$.
\end{observation}


\begin{observation}\label{obs2}
Two sets of Kraus operators $\{T_x\}$ and $\{T_x'\}$ for the same channel ${\cal T}$,
give rise to two distinct measurement processes (POVMs $\{T^\dag_xT_x\}$ and $\{T_x'^\dag T_x'\}$ respectively)
if they cannot be obtained one from another as $\{T_x'\}=U\{T_x\}$, with a unitary $U$  that is diagonal or it is a permutation
(or combination of the two).
This introduces an equivalence relation between Kraus decompositions that will be useful in the following.
Specifically, we will not distinguish measurements (Kraus decompositions) by elements of $U(d)$,
but rather by elements of the quotient group $U(d)/ (U(1) \wr S_d)$.
Here $U(1) \wr S_d$ is the wreath product of the group $U(1)$ by the symmetric group $S_d$.
This is the generalized permutation group where the entries are elements of the unitary group $U(1)$.
Clearly $U(1) \wr S_d$ is a subgroup of $U(d)$.
\end{observation}


\section{Qubit channels}

Let us consider a qubit channel  ${\cal T}:{\mathfrak{S}}(\mathbb{C}^2)\to {\mathfrak{S}}(\mathbb{C}^2)$.
Any element of the set $\mathfrak{S}(\mathbb{C}^2)$ can be written as $\rho=\frac{1}{2}\left(I + \boldsymbol{r} \cdot \boldsymbol{\sigma}\right)$ with $\boldsymbol{r}\in\mathbb{R}^3$,
$\|\boldsymbol{r}\|\leq 1$,
and $\boldsymbol{\sigma}$ the vector of Pauli operators.
The channel ${\cal T}$ can always be represented as an affine transformation
in $\mathbb{R}^3$,
\begin{equation}
\boldsymbol{r}\stackrel{{\cal T}}{\longrightarrow}
\boldsymbol{T} \, \boldsymbol{r} + \boldsymbol{t},
\end{equation}
where $\boldsymbol{T}$ is a $3\times 3$ real matrix and $ \boldsymbol{t}=(t_1,t_2,t_3)^\top\in \mathbb{R}^3$.

By means of orthogonal transformations, corresponding to unitary pre- and post-processing to the channel $\cal T$,
the matrix $\boldsymbol{T}$ can be brought to diagonal form. Then, let us define
\begin{equation}
\mathfrak{T}:=\left\{ {\cal T}:{\mathfrak{S}}(\mathbb{C}^2)\to {\mathfrak{S}}(\mathbb{C}^2)\,  | \,
{\cal T} \; \text{is linear and CPTP and} \; {\cal T} \leftrightarrow ( {\boldsymbol{T}}=diag, {\boldsymbol{t}}) \right\}.
\end{equation}

When $\cal T$ belongs to the closure of extreme points of $\mathfrak{T}$, it can be parametrized by
(see e.g. \cite{OMR23})
\begin{equation}\label{Mt}
\boldsymbol{T}=\begin{pmatrix}
\cos(\phi-\theta) & 0 & 0 \\ 0 & \cos(\phi+\theta) & 0 \\ 0 & 0 & \cos(\phi-\theta)\cos(\phi+\theta)
\end{pmatrix},\quad
\boldsymbol{t}=\begin{pmatrix}
0 \\ 0 \\ \sin(\phi-\theta)\sin(\phi+\theta)
\end{pmatrix},
\end{equation}
with $\theta,\phi\in[0,\pi]$.
In term of Kraus operators this reads
\begin{equation}
{\cal T}(\rho)=\sum_{x=0}^1 T_x\rho T_x^\dag,
\end{equation}
with $T_0,T_1$ expressed in the canonical basis as
\begin{equation}\label{kmatrix}
T_0=\begin{pmatrix}
\cos\theta && 0\\
0 && \cos\phi
\end{pmatrix}, \quad
T_1=\begin{pmatrix}
0 && \sin{ \phi} \\
\sin{\theta}  && 0
\end{pmatrix}.
\end{equation}

Any other quantum channel ${\cal T}$ that is not on the closure of extreme points of $\mathfrak{T}$,
can be written as convex combination
of two maps in the closure of extreme points \cite{OMR23}, namely as
\begin{equation}\label{Tgen}
{\cal T}=\lambda {\cal T}_{\theta,\phi}+(1-\lambda) {\cal T}_{\theta',\phi'}, \qquad 0<\lambda <1,
\end{equation}
where ${\cal T}_{\theta,\phi}$ and ${\cal T}_{\theta',\phi'}$ are in the closure of the extreme points of $\mathfrak{T}$.
Then,

\begin{equation}\label{Tconvex}
{\cal T}(\rho)=\lambda \left(T_0\rho T_0^\dag+T_1\rho T_1^\dag\right)
+(1-\lambda) \left(T'_0\rho {T'_0}^\dag+T'_1\rho {T'_1}^\dag\right).
\end{equation}
with Kraus operators
\begin{equation}\label{Kconvex}
\sqrt{\lambda} T_0, \quad
\sqrt{\lambda} T_1, \quad
\sqrt{1-\lambda} T'_0, \quad
\sqrt{1-\lambda} T'_1,
\end{equation}


\begin{proposition}\label{T2}
For any qubit channel ${\cal T}$ of Kraus rank 2, the optimal measurement on single feedback step is given by
the Kraus operators \eqref{kmatrix}.
\end{proposition}

\begin{proof}
We can construct the optimal Kraus operators starting from those of Eq.\eqref{kmatrix} as
\begin{equation}\label{KrausTt}
\widetilde{T_i}=\sum_{j=0}^1 U_{ij} T_j,
\end{equation}
by using the unitary $2\times 2$ matrix
\begin{equation}\label{U22}
U=\begin{pmatrix}
\cos\alpha & \sin\alpha \\
-\sin\alpha & \cos\alpha
\end{pmatrix}.
\end{equation}
Note that we do not need to consider a more general unitary as consequence of Observation \ref{obs2}.

Eq.\eqref{KrausTt} implies that{
\begin{align}\label{FAB}
F_1=\frac{1}{2} (1+\left|\cos (\theta- \phi) \sin^2\alpha-\sin \theta  \sin \phi\right| +\left| \cos (\theta- \phi)\cos^2\alpha -\sin \theta  \sin \phi\right| ),
\end{align}}

It is easy to see that the maximum of \eqref{FAB} over $\alpha$ occurs for $\alpha=0,\pi/2,\pi$, i.e. for the Kraus decomposition with $T_0,T_1$.
\end{proof}

\medskip

When the qubit channel has Kraus rank $>2$, we cannot derive a general expression for the optimal measurement. However, we can state the following result for any qubit channel.

\begin{proposition}\label{persi}
For any qubit channel ${\cal T}$ the optimal measurement on the environment when adopting the Bayesian strategy is the same at each feedback step.
\end{proposition}

\begin{proof}
Let $\{\tilde{T}_{x}\}$ be the optimal measurement at the $k$-step of the Bayesian strategy.
We will show that any measurement $\{T_x\}$ at the $k+1$-th step cannot give a higher fidelity.

By hypothesis, it is
\begin{equation}
\sum_{x_1,\ldots,x_k}\left({\rm{tr}} |{\tilde{T}_{x_k}}|  \tilde{A}_{x_{k-1},\ldots,x_1}... ||\right)^2
\geq\sum_{x_1,\ldots,x_k}\left({\rm{tr}}|{{T}_{x_k}}|  \tilde{A}_{x_{k-1},\ldots,x_1}...||\right)^2,
\end{equation}
where $\tilde{A}_{x_{k-1},\ldots,x_1}$
stands for the combination of Kraus operators of previous $k-1$ steps.

The following relation
\begin{equation}
  ({\rm tr}M)^2={\rm tr}M^2+2\det M,
\end{equation}
valid for any $2\times 2$ matrix $M$, yields
\begin{align}
&\sum_{x_1,\ldots,x_k}\left({\rm{tr}} |{\tilde{T}_{x_k}}| \, \tilde{A}_{x_{k-1},\ldots,x_1}...||\right)^2
\geq\sum_{x_1,\ldots,x_k}\left({\rm{tr}}|{{T}_{x_k}}| \, \tilde{A}_{x_{k-1},\ldots,x_1}...||\right)^2 \label{1line}\\
\Longleftrightarrow
& \sum_{x_1,\ldots,x_k} {\rm{tr}} \left(|{\tilde{T}_{x_k}}| \, \tilde{A}_{x_{k-1},\ldots,x_1}...||\right)^2
+2  \sum_{x_1,\ldots,x_k} \det |{\tilde{T}_{x_k}}| \det |\tilde{A}_{x_{k-1},,\ldots,x_1}| \notag\\
&\geq\sum_{x_1,\ldots,x_k}{\rm{tr}}\left(|{{T}_{x_k}}| \, \tilde{A}_{x_{k-1},\ldots,x_1}...||\right)^2
+2  \sum_{x_1,\ldots,x_k} \det |T_{x_k}| \det |\tilde{A}_{x_{k-1},,\ldots,x_1}|\label{2line}\\
\Longleftrightarrow
&  \sum_{x_k} \det |{\tilde{T}_{x_k}}|
\geq  \sum_{x_k} \det |T_{x_k}|.\label{3line}
\end{align}
From \eqref{1line} to \eqref{2line} we took into account that $\det|M_1|M_2||=\det|M_1|\det|M_2|$.
From  \eqref{2line} to \eqref{3line} we used the fact that
$\sum_{x_1,\ldots,x_k} \left(|{\tilde{T}_{x_k}}| \, \tilde{A}_{x_{k-1},\ldots,x_1}...||\right)^2=
\sum_{x_1,\ldots,x_k} \left(|{{T}_{x_k}}| \, \tilde{A}_{x_{k-1},\ldots,x_1}...||\right)^2=I$.

Now, at the $k+1$-th step, we have
\begin{align}
\sum_{x_1,\ldots,x_k,x_{k+1}}\left({\rm{tr}} |T_{x_{k+1}} \, |{\tilde{T}_{x_k}}|  \tilde{A}_{x_{k-1},\ldots,x_1}...|||\right)^2&=
\sum_{x_1,\ldots,x_k,x_{k+1}}{\rm{tr}} \left(|T_{x_{k+1}} \, |{\tilde{T}_{x_k}}|  \tilde{A}_{x_{k-1},\ldots,x_1}...|||\right)^2 \notag\\
&+2  \sum_{x_1,\ldots,x_k,x_{k+1}} \det  |T_{x_{k+1}}| \det |{\tilde{T}_{x_k}}| \tilde{A}_{x_{k-1},,\ldots,x_1}...|| \\
&\leq \sum_{x_1,\ldots,x_k,x_{k+1}}{\rm{tr}} \left(|\tilde{T}_{x_{k+1}} \, |{\tilde{T}_{x_k}}| \tilde{A}_{x_{k-1},\ldots,x_1}...|||\right)^2 \notag\\
&+2  \sum_{x_1,\ldots,x_k,x_{k+1}} \det  |\tilde{T}_{x_{k+1}}| \det |{\tilde{T}_{x_k}}| \tilde{A}_{x_{k-1},,\ldots,x_1}...|| \\
&=\sum_{x_1,\ldots,x_k,x_{k+1}}\left({\rm{tr}} |\tilde{T}_{x_{k+1}} |{\tilde{T}_{x_k}}| \tilde{A}_{x_{k-1},\ldots,x_1}...|||\right)^2.
\end{align}
The inequality is obtained using the fact that
$\sum_{x_1,\ldots,x_k,x_{k+1}} \left(|{T}_{x_{k+1}}|{\tilde{T}_{x_k}}| \, \tilde{A}_{x_{k-1},\ldots,x_1}...|||\right)^2=$ \newline
$\sum_{x_1,\ldots,x_k,x_{k+1}} \left(|\tilde{T}_{x_{k+1}}|...{\tilde{T}_{x_k}}| \, \tilde{A}_{x_{k-1},\ldots,x_1}...||\right)^2=I$,
together with \eqref{3line}. 
\end{proof}


\begin{remark}
As consequence of Proposition \ref{persi} we can deduce the inutility of Bayesian feedback (with respect to Markovian feedback) for qubit channel of Kraus rank 2. This is because the optimal measurement, according to Proposition \ref{T2}, gives Kraus operators with commuting absolute values (see Observation \ref{obs1}).
\end{remark}


\begin{remark}
Proposition \ref{persi} states that for Bayesian feedback it is enough to find the optimal measurement at the first step.
In general this is not the case for Markovian feedback.\footnote{Essentially because
$\sum_{x_1,\ldots,x_k} \left(|{{T}_{x_k}}| \ldots |{T}_{x_{1}}|\right)^2\neq I$.}
Here, if the optimal set $\{\tilde{T}_{x_1}\}$ is unique (in the sense of Observation \ref{obs2}),
then it keeps going on at each step (as in the case of channel with Kraus rank equal to 2).
However, when the Kraus rank of the channel is greater than 2, it might happen that such a set is not unique, or in other words there might be several unitaries that applied to the Kraus decomposition \eqref{Kconvex} give the same value of $F_1$. In turn, these unitaries can give different values of $F_2$ (and consequently $F_n$), while keeping the same value of $F_2'$. Thus, for the Markovian strategy, one needs to identify among them the unitary that maximize $F_2$.
\end{remark}

\begin{example}
A simple channel of Kraus rank equal to 3 can be build by considering $\theta=\theta'=\phi'=0$ and $\phi=\frac{\pi}{2}$ in Eq.\eqref{Kconvex}. Its Kraus operators read
\begin{equation}\label{Kex}
T_1=\sqrt{\lambda} \begin{pmatrix}
1 & 0\\
0 & 0
\end{pmatrix},
\quad
T_2=\sqrt{\lambda} \begin{pmatrix}
0 & 1\\
0 & 0
\end{pmatrix},
\quad
T_3=\sqrt{1-\lambda} \begin{pmatrix}
1 & 0\\
0 & 1
\end{pmatrix}.
\end{equation}
Note that $F_1$ computed with \eqref{Kex} results
\begin{equation}
F_1=\frac{1}{2}+\frac{1-\lambda}{2}.
\end{equation}
According to \cite{R97} a general $3\times 3$ unitary can be written as\footnote{We get rid of diagonal pre and post matrices containing uni-modular factors because irrelevant for our purposes.}
\begin{equation}
U=\begin{pmatrix}
1 & & \\
 & c_{23} & s_{23} \\
 & -s_{23} & c_{23}
\end{pmatrix}
\begin{pmatrix}
e^{-i\delta} & & \\
& 1 & \\
& & 1
\end{pmatrix}
\begin{pmatrix}
 c_{13} & & s_{13} \\
  & 1 & \\
  -s_{13} & & c_{13}
\end{pmatrix}
\begin{pmatrix}
e^{i\delta} & & \\
& 1 & \\
& & 1
\end{pmatrix}
\begin{pmatrix}
 c_{12} & s_{12} & \\
 -s_{12} & c_{12} & \\
& & 1 \\
\end{pmatrix},
\end{equation}
where $c_{ij},s_{ij}$ stand for $\cos\theta_{ij}$ and $\sin\theta_{ij}$ respectively.
Now, considering that the linear combination of $T_1$ and $T_2$ does not improve the fidelity (as a consequence of  Proposition \ref{T2}), we discard the last term (and hence the pre-last by Observation \ref{obs2}). We end up considering a unitary
\begin{equation}
U=\begin{pmatrix}
1 & & \\
 & c_{23} & s_{23} \\
 & -s_{23} & c_{23}
\end{pmatrix}
\begin{pmatrix}
e^{-i\delta} & & \\
& 1 & \\
& & 1
\end{pmatrix}
\begin{pmatrix}
 c_{13} & & s_{13} \\
  & 1 & \\
  -s_{13} & & c_{13}
\end{pmatrix}
=\begin{pmatrix}
e^{-i\delta} c_{13} & 0 & e^{-i\delta}s_{13} \\
-s_{13}s_{23}  & c_{23} & c_{13}s_{23} \\
-c_{23}s_{13} & -s_{23}& c_{13}c_{23}
\end{pmatrix}.
\end{equation}
At this point we note that also the parameter $e^{-i\delta}$ can be neglected,
because it factors out in the first Kraus operator.
Thus we have
\begin{align}
\widetilde{T}_1&=c_{13}T_1+s_{13}T_3,\\
\widetilde{T}_2&=-s_{13}s_{23}T_1+c_{23}T_2+c_{13}s_{23}T_3,\\
\widetilde{T}_3&=-s_{13}c_{23}T_1-s_{23}T_2+c_{13}c_{23}T_3. 
\end{align}
From these Kraus operators we obtain the maximum single step fidelity
\begin{equation}\label{F1max3K}
F_1=\frac{1}{2}+\frac{\sqrt{1-\lambda}}{2},
\end{equation}
when
\begin{equation}
\theta_{13}=k\frac{\pi}{2}-\frac{1}{2}\arcsin\sqrt{\lambda},
\; k\in\mathbb{Z},
\qquad
\theta_{23}=\textrm{any value}.
\end{equation}
Thus, considering that $\theta_{23}$ can take any value, there will be infinitely many Kraus decompositions
giving the fidelity \eqref{F1max3K} (even accounting for the Observation \ref{obs2}).
Moving on to the second feedback step, all of them give
\begin{equation}
F'_2=\frac{1}{2}+\frac{\left(\sqrt{1-\lambda}\right)^2}{2},
\end{equation}
but it results
\begin{equation}
F_2=\frac{1}{2}+\frac{\left(\sqrt{1-\lambda}\right)^2}{2},
\end{equation}
only when $\theta_{23}$ equals integer multiples of $\frac{\pi}{2}$ (while for other values of $\theta_{23}$ one gets a smaller value of $F_2$).
Generalizing, under the same conditions, we get
\begin{equation}
F_n=F_n'=\frac{1}{2}+\frac{\left(\sqrt{1-\lambda}\right)^n}{2}.
\end{equation}
\end{example}

\medskip

We have numerically verified that other rank 3  qubit channels (coming from covex combination of \eqref{kmatrix} and identity) and rank 4 qubit channels (as from Eq.\eqref{Kconvex})
allows for the same performance of Bayesian and Markovian feedback.

Actually we created a grid in the parameter space $\{\lambda, \theta, \phi, \theta', \phi'\}$ with step 0.05 for $\lambda$ and
$\frac{\pi}{50}$ for the angles. Then, on each point of this grid we evaluated
\begin{align}
F_2=\max_{\{U\}}
\frac{1}{4}\sum_{x_1,x_2} \left( {\rm tr}
\left|\widetilde{T}_{x_2}\right| \,
\left|\widetilde{T}_{x_1}\right|
\right)^2,
\qquad
\text{and}
\qquad
F_2'=\max_{\{U\}}
\frac{1}{4}\sum_{x_1,x_2} \left( {\rm tr}
\left|\widetilde{T}_{x_2} \,
\left|\widetilde{T}_{x_1}\right|\right|
\right)^2,
\end{align}
with
\begin{align}
\begin{pmatrix}
\tilde{T}_0\\
\tilde{T}_1\\
\tilde{T}'_0\\
\tilde{T}'_1
\end{pmatrix}=U \begin{pmatrix}
{T}_0\\
{T}_1\\
{T}'_0\\
{T}'_1
\end{pmatrix}.
\end{align}

At each step,
we have generated $10^5$ random unitary operators $U$ by sampling from $U(4)$ according to the Haar measure.
Specifically, this has been done by using the Ginibre ensemble \cite{Gin} consisting of matrices whose entries are independent and identically distributed standard normal complex random variables
and then ortho-normalizing such matrices \cite{Mez}.

Since already at the second step we have $F_2=F_2'$ and the solution $U$ achieving this is unique (in the sense of Observation \ref{obs2}), the result keeps going on for any step $n$.

Therefore, we are led to the following conjecture.


\begin{conjecture}
For any qubit channel it is $F_n'=F_n$, $\forall n$.
\end{conjecture}



\section{Qutrit channels}

In this Section we consider two enlightening examples of channels
${\cal T}:{\mathfrak{S}}(\mathbb{C}^3)\to {\mathfrak{S}}(\mathbb{C}^3)$.
The first one (dephasing channel) for which Bayesian feedback does not offer any advantage over the Markovian feedback, and the second one (amplitude damping) which instead does.

\subsection{Dephasing channel}

According to Ref.\cite{RMM22}, the action of the dephasing channel on a qutrit, with dephasing parameter $\gamma\geq 0$,
should be written as:
\begin{equation}
{\cal D}_{\gamma}(\rho)=\sum_{j=0}^{\infty} D_j\rho D_j^\dagger,
\end{equation}
where the Kraus operators read
\begin{equation}
D_j=e^{-\frac{1}{2}\gamma H^2}\frac{(-i \sqrt{\gamma} H )^j}{\sqrt{j!}},
\end{equation}
with $H$ the Hamiltonian of the system. In the basis $\{ |0\rangle,|1\rangle, |2\rangle \}$ composed by the eigenvectors of $H$, we have
\begin{equation}\label{rhoFock}
\rho=\sum_{m,n=0}^{2} \rho_{m,n}\ket m \bra n \mapsto
{\cal D}_\gamma \left(\rho\right)=\sum_{m,n=0}^{2} e^{-\frac{1}{2}\gamma \left(m-n\right)^2} \rho_{m,n}\ket m \bra n.
\end{equation}

However it is known that this map can be obtained with a maximum of $9$ Kraus operators,
instead of infinitely many. Actually, it is possible to get it with just 3:
\begin{equation*}
\widetilde{D}_j=\sum_{k=0}^2 U_{jk}D_k,
\end{equation*}
through a suitable unitary $U$.
Explicitly, setting $q\equiv e^{-\gamma/2}$, we have
\begin{align}\label{Kdeph}
\widetilde{D}_0=\sqrt{\frac{1-q^4}{2}}
\begin{pmatrix}
-1 & 0 & 0 \\
0 & 0 & 0 \\
0 & 0 & 1
\end{pmatrix}, \qquad
\widetilde{D}_{\tiny\begin{matrix}
1 \\ 2 \end{matrix}}=\sqrt{\frac{2+q^4\pm \sqrt{8q^2+q^8}}{2(2+\alpha^2_\pm)}}
\begin{pmatrix}
1 & 0 & 0 \\
0 & \alpha_\pm & 0 \\
0 & 0 & 1
\end{pmatrix},
\end{align}
represented in the basis $\{ |0\rangle,|1\rangle, |2\rangle \}$ and having defined
\begin{align}
\alpha_\pm\equiv \frac{q^3\sqrt{8+q^6}\pm(4-q^6)}{\sqrt{8+q^6}\pm 3q^3}.
\end{align}
Being the Kraus operators \eqref{Kdeph} diagonal, also their absolute value will be diagonal and hence commuting.
Actually, any unitary transformation of the Kraus operators \eqref{Kdeph} will lead to diagonal operators.
As a consequence the Bayesian feedback will not be effective for this channel (see Observation \ref{obs1}).


\subsection{Amplitude damping channel}

According to Ref. \cite{MCM11}, the action of the amplitude damping channel on a qutrit, with damping parameter $p\in[0,1]$,
should be written as:
\begin{align}
{\cal A}(\rho)=\sum_{k=0}^2 A_k\rho A_k^\dag,
\end{align}
where the Kraus operators are
\begin{align}
A_0&=|0\rangle\langle 0|+\sqrt{1-p} |1\rangle\langle 1|+(1-p) |2\rangle\langle 2|, \\
A_1&=\sqrt{p} |0\rangle\langle 1|+\sqrt{2p(1-p)} |1\rangle\langle 2|, \\
A_2&=p |0\rangle\langle 2|.
\end{align}
In the canonical basis $\{|0\rangle,|1\rangle,|2\rangle\}$
\begin{align}\label{Abare}
A_0=\begin{pmatrix}
1 & 0 & 0 \\
0 & \sqrt{1-p} & 0 \\
0 & 0 & (1-p)
\end{pmatrix}, \quad
A_1=\begin{pmatrix}
0 & \sqrt{p} & 0 \\
0 & 0 & \sqrt{2p(1-p)} \\
0 & 0 & 0
\end{pmatrix},\quad
A_2=\begin{pmatrix}
0 & 0 & p\\
0 & 0 & 0 \\
0 & 0 & 0
\end{pmatrix}.
\end{align}
By polar decomposition, it results
\begin{align}\label{absA}
|A_0|=\begin{pmatrix}
1 & 0 & 0 \\
0 & \sqrt{1-p} & 0 \\
0 & 0 & (1-p)
\end{pmatrix}, \quad
|A_1|=\begin{pmatrix}
0 & 0 & 0 \\
0 & \sqrt{p} & 0 \\
0 & 0 & \sqrt{2p(1-p)}
\end{pmatrix},\quad
|A_2|=\begin{pmatrix}
0 & 0 & 0 \\
0 & 0 & 0 \\
0 & 0 & p
\end{pmatrix}.
\end{align}
They are commuting, hence, according to Observation \ref{obs1}, we should conclude that Bayesian feedback is not effective also in this case.

However, in Ref.\cite{MCM11} it was pointed out that the best Kraus decomposition of this channel for $F_1$
should come with a unitary
\begin{equation}
U=\begin{pmatrix}
\frac{e^{i\theta}}{\sqrt{2}} &  & \frac{e^{i\phi}}{\sqrt{2}}\\
 & 1 & \\
\frac{-e^{-i\phi}}{\sqrt{2}} &  & \frac{e^{i\theta}}{\sqrt{2}}
\end{pmatrix},
\end{equation}
applied to the Kraus operators \eqref{Abare}.
According to Observation \ref{obs2} we can safely set $\theta=\phi=0$, getting (in the canonical basis)
\begin{align}\label{Atilde}
\widetilde{A}_0=\frac{1}{\sqrt{2}}\begin{pmatrix}
1 & 0 & p \\
0 & \sqrt{1-p} & 0 \\
0 & 0 & (1-p)
\end{pmatrix}, \quad
\widetilde{A}_1=\begin{pmatrix}
0 & \sqrt{p} & 0 \\
0 & 0 & \sqrt{2p(1-p)} \\
0 & 0 & 0
\end{pmatrix},\quad
\widetilde{A}_2=\frac{-1}{\sqrt{2}}\begin{pmatrix}
1 & 0 & -p\\
0 & \sqrt{1-p} & 0 \\
0 & 0 & (1-p)
\end{pmatrix}.
\end{align}
By polar decomposition, it follows
\begin{align}
|\widetilde{A}_0|&=\frac{1}{\sqrt{2}}\begin{pmatrix}
\frac{ \left(\alpha+\overline{p}\right)^2 \sqrt{\overline{p} +p(p+\alpha)}}{ 1+\left(\alpha+\overline{p}\right)^2 }
+\frac{ \left(\alpha-\overline{p}\right)^2 \sqrt{\overline{p} +p(p-\alpha)}}{ 1+\left(\alpha-\overline{p}\right)^2 }
& 0 &
\frac{-\sqrt{ \overline{p} +p(p-\alpha)} + \sqrt{ \overline{p} +p(p+\alpha)}}{2\alpha} \\
0 & \sqrt{1-p} & 0 \\
\frac{-\sqrt{ \overline{p} +p(p-\alpha)} + \sqrt{ \overline{p} +p(p+\alpha)}}{2\alpha} & 0 &
\frac{\sqrt{\overline{p}+p(p+\alpha)}}{1+(\alpha+\overline{p})^2}
+\frac{\sqrt{\overline{p}+p(p-\alpha)}}{1+(\alpha-\overline{p})^2}
\end{pmatrix},\label{absA0t} \\
|\widetilde{A}_1|&=\begin{pmatrix}
0 & 0 & 0 \\
0 & \sqrt{p} & 0 \\
0 & 0 & \sqrt{2p(1-p)}
\end{pmatrix},  \label{absA1t}\\
|\widetilde{A}_2|&=\frac{1}{\sqrt{2}}\begin{pmatrix}
\frac{ \left(\alpha+\overline{p}\right)^2 \sqrt{\overline{p} +p(p+\alpha)}}{ 1+\left(\alpha+\overline{p}\right)^2 }
+\frac{ \left(\alpha-\overline{p}\right)^2 \sqrt{\overline{p} +p(p-\alpha)}}{ 1+\left(\alpha-\overline{p}\right)^2 }
& 0 &
\frac{\sqrt{ \overline{p} +p(p-\alpha)} - \sqrt{ \overline{p} +p(p+\alpha)}}{2\alpha} \\
0 & \sqrt{1-p} & 0 \\
\frac{\sqrt{ \overline{p} +p(p-\alpha)} - \sqrt{ \overline{p} +p(p+\alpha)}}{2\alpha} & 0 &
\frac{\sqrt{\overline{p}+p(p+\alpha)}}{1+(\alpha+\overline{p})^2}
+\frac{\sqrt{\overline{p}+p(p-\alpha)}}{1+(\alpha-\overline{p})^2}
\end{pmatrix},\label{absA2t}
\end{align}
where
\begin{align}
\overline{p}\equiv 1-p, \qquad \alpha\equiv\sqrt{2+p(p-2)}\equiv\sqrt{1+\overline{p}^2}.
\end{align}


Assuming to perform the same measurement at each step, we have
\begin{align}
F_n&=\frac{1}{9}\sum_{x_1,\ldots,x_n} \left( {\rm tr} \left| \widetilde{A}_{x_n}\right| \ldots \left|\widetilde{A}_{x_1}\right|
\right)^2,\\
F'_n&=
\frac{1}{9}\sum_{x_1,\ldots, x_n} \left( {\rm tr} \left|\widetilde{A}_{x_n} \ldots
\left|\widetilde{A}_{x_2} \,
\left|\widetilde{A}_{x_1}\right|\right|\right|
\right)^2.
\end{align}

In Fig.\ref{fig1} we show the difference $F'_n-F_n$ vs the damping parameter $p$ and the number $n$ of feedback steps.
Its positiveness
shows better performance of Bayesian feedback with respect to Markovian feedback.
The range of values of $p$ for which the quantity is significantly greater than zero, enlarges as $n$ increases.
As function of $n$ the advantage of Bayesian feedback, while initially increasing, tends later on to saturate.

\begin{figure}[H]
\begin{center}
          \includegraphics[width=0.65\textwidth]{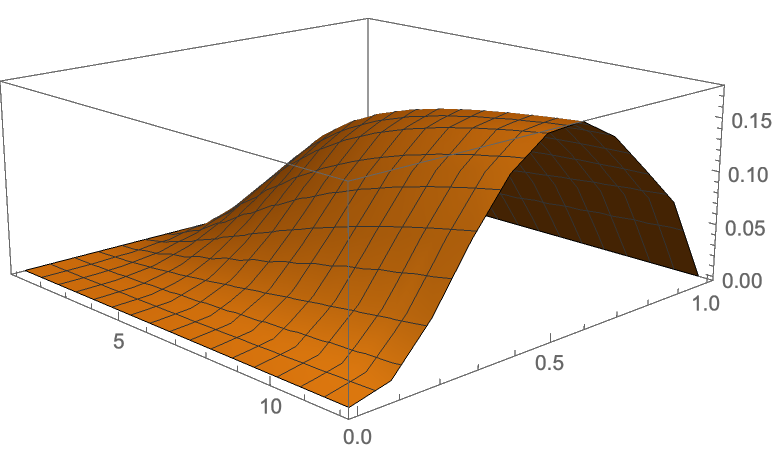}
          \caption{Difference $F'_n-F_n$ vs damping parameter $p$ and number $n$ of feedback steps.}
\label{fig1}
\end{center}
\end{figure}

Since the Kraus decomposition \eqref{Atilde} is unique (in the sense of Observation \ref{obs2}), the measurement will be the same at each step concerning the Markovian strategy.

However, we were not fully justified to assume the same measurement at each step for Bayesian feedback, as Proposition \ref{persi} does not extends to qutrit channels.
Thus, we have numerically studied the quantity
\begin{align}
F'_n=\max_{\{U^{(k)}\}}
\frac{1}{9}\sum_{x_1,\ldots, x_n} \left( {\rm tr} \left|\widetilde{A}^{(n)}_{x_n} \ldots
\left|\widetilde{A}^{(2)}_{x_2} \,
\left|\widetilde{A}^{(1)}_{x_1}\right|\right|\right|
\right)^2,
\end{align}
with
\begin{equation}
\left\{\begin{array}{lll}
\widetilde{A}^{(k)}_{x_k}=\sum_{y_k}U^{(k)}_{x_ky_k}\widetilde{A}^{(1)}_{y_k}, & & k>1\\ \\
\widetilde{A}^{(1)}_{x_1}\equiv \widetilde{A}_{x_1}
\end{array}\right..
\end{equation}

We created a grid for $p\in[0,1]$ with step 0.05. For every point in such a grid,
at each step, $k>1$,
we have generated $10^5$ random unitary operators $U^{(k)}$ by sampling from $U(3)$ according to the Haar measure.
Also in this case we have used the Ginibre ensemble \cite{Gin} consisting of matrices whose entries are independent and identically distributed standard normal complex random variables
and then ortho-normalizing such matrices \cite{Mez}.

The numerics show that the Kraus decomposition (measurement on the environment) must be the same at every step.


\section{Conclusion}

We studied discrete time feedback aimed at reclaiming quantum information after a channel action.
We compared Bayesian and Markovian strategies. It turns out that Bayesian feedback performs no better than Markovian feedback when used on qubit channels.
However, its superior performance appear in higher dimensional channels like e.g. the qutrit amplitude damping channel.

There are a number of problems that this study poses in future perspective.
First, find other examples of channels for which Bayesian feedback is more effective than Marokovian one.
Then, for these channels see if by increasing the dimension the advantage increases.
Last but not least, it would be desirable
to study the continuum limit of our discrete feedback model.


\section*{Acknowledgments}
The authors M. R. and S. M. acknowledge financial support from
``PNRR MUR project PE0000023-NQSTI".





\end{document}